\documentclass[a4paper,UKenglish,cleveref,autoref]{lipics-v2019}

\title{Compiling Crossing-free Geometric Graphs with Connectivity Constraint for Fast Enumeration, Random Sampling, and Optimization}

\titlerunning{Compiling Crossing-free Geometric Graphs with Connectivity Constraint}

\author{Yu Nakahata}{Graduate School of Informatics, Kyoto University, Japan}{nakahata.yu.27e@st.kyoto-u.ac.jp}{https://orcid.org/0000-0002-8947-0994}{This work was supported by JSPS KAKENHI Grant Number JP19J21000.}
\author{Takashi Horiyama}{Graduate School of Information Science and Technology, Hokkaido University, Japan}{horiyama@ist.hokudai.ac.jp}{https://orcid.org/0000-0001-9451-259X}{This work was supported by JSPS KAKENHI Grant Numbers JP15H05711 and JP18K11153.}
\author{Shin-ichi Minato}{Graduate School of Informatics, Kyoto University, Japan}{minato@i.kyoto-u.ac.jp}{https://orcid.org/0000-0002-1397-1020}{This work was supported by JSPS KAKENHI Grant Number JP15H05711.}
\author{Katsuhisa Yamanaka}{Faculty of Science and Engineering, Iwate University, Japan}{yamanaka@cis.iwate-u.ac.jp}{https://orcid.org/0000-0002-4333-8680}{This work was supported by JSPS KAKENHI Grant Numbers JP18H04091 and JP19K11812.}

\authorrunning{Y. Nakahata, T. Horiyama, S. Minato, and K. Yamanaka}

\Copyright{Yu Nakahata, Takashi Horiyama, Shin-ichi Minato and Katsuhisa Yamanaka}



\category{}

\relatedversion{}

\supplement{}

\acknowledgements{}

\nolinenumbers 

\hideLIPIcs  

\EventEditors{John Q. Open and Joan R. Access}
\EventNoEds{2}
\EventLongTitle{42nd Conference on Very Important Topics (CVIT 2016)}
\EventShortTitle{CVIT 2016}
\EventAcronym{CVIT}
\EventYear{2016}
\EventDate{December 24--27, 2016}
\EventLocation{Little Whinging, United Kingdom}
\EventLogo{}
\SeriesVolume{42}
\ArticleNo{23}

\usepackage[linesnumbered,ruled,vlined]{algorithm2e}
\SetKwInOut{Input}{input}
\SetKwInOut{Output}{output}
\usepackage{amsmath}
\usepackage{amssymb}
\usepackage{amsfonts}
\usepackage{amsthm}
\usepackage{stmaryrd}
\usepackage{url}
\usepackage{tikz}
\usetikzlibrary{
  positioning,
  calc,
}
\usepackage{lipsum}
\usepackage[shortlabels]{enumitem}

\newcommand{\func}[3]{#1 \colon #2 \to #3}
\newcommand{\size}[1]{\left| #1 \right|}
\newcommand{\bigO}[1]{\mathrm{O}(#1)}
\newcommand{\ostar}[1]{\mathrm{O}^{*}(#1)}
\newcommand{\mf}[1]{\mathfrak{#1}}

\newcommand{\set}[1]{\left\{ #1 \right\}}
\newcommand{\inset}[2]{\left\{ #1 \;\middle|\; #2 \right\}}

\newcommand{\pts}[1]{\mathrm{pts}(#1)}
\newcommand{\lft}[1]{\mathrm{lft}(#1)}
\newcommand{\rgt}[1]{\mathrm{rgt}(#1)}
\newcommand{\low}[1]{\mathrm{low}(#1)}
\newcommand{\upp}[1]{\mathrm{upp}(#1)}
\newcommand{\rex}[1]{\mathrm{rex}(#1)}
\newcommand{\depend}{\sqsubset}

\newcommand{\sgm}{\mathcal{S}_P}
\newcommand{\cfg}{\mathfrak{C}^{\mathsf{cf}}_P}
\newcommand{\str}{\mathfrak{C}^{\mathsf{st}}_P}
\newcommand{\prestr}{\mathfrak{D}^{\mathsf{st}}_P}
\newcommand{\scy}{\mathfrak{C}^{\mathsf{sc}}_P}
\newcommand{\prescy}{\mathfrak{D}^{\mathsf{sc}}_P}
\newcommand{\dsg}{\vec{\mathcal{S}}_P}
\newcommand{\dsc}{\vec{\mathfrak{C}}^{\mathsf{sc}}_P}
\newcommand{\predsc}{\vec{\mathfrak{D}}^{\mathsf{sc}}_P}

\newcommand{\catl}[1]{\mathrm{Cat}(#1)}

\newcommand{\cnt}[1]{\mathrm{cnt}(#1)}
\newcommand{\degr}[2]{\mathrm{deg}(#1, #2)}
\newcommand{\indegr}[2]{d^{\,\mathrm{in}}(#1, #2)}
\newcommand{\outdegr}[2]{d^{\,\mathrm{out}}(#1, #2)}

\begin{document}

\maketitle

\begin{abstract}
Given $n$ points in the plane, we propose algorithms to compile \emph{connected} crossing-free geometric graphs into directed acyclic graphs (DAGs).
The DAGs allow efficient counting, enumeration, random sampling, and optimization.
Our algorithms rely on Wettstein's framework to compile several crossing-free geometric graphs.
One of the remarkable contributions of Wettstein is to allow dealing with geometric graphs with ``connectivity'',
since it is known to be difficult to efficiently represent geometric graphs with such \emph{global} property.
To achieve this, Wettstein proposed specialized techniques for crossing-free spanning trees and crossing-free spanning cycles and
invented compiling algorithms running in $\ostar{7.044^n}$ time and $\ostar{5.619^n}$ time, respectively.

Our first contribution is to propose a technique to deal with the connectivity constraint more simply and efficiently.
It makes the design and analysis of algorithms easier, and yields improved time complexity.
Our algorithms achieve $\ostar{6^n}$ time and $\ostar{4^n}$ time
for compiling crossing-free spanning trees and crossing-free spanning cycles, respectively.
As the second contribution, we propose an algorithm to optimize the area surrounded by crossing-free spanning cycles.
To achieve this, we modify the DAG so that it has additional information.
Our algorithm runs in $\ostar{4.829^n}$ time to find an area-minimized (or maximized) crossing-free spanning cycle of a given point set.
Although the problem was shown to be NP-complete in 2000, as far as we know, there were no known algorithms faster than the obvious $\ostar{n!}$ time algorithm for 20 years.

\end{abstract}

\section{Introduction}
Let $P$ be a set of $n$ points in the plane.
We assume $P$ to be in general position, that is, no three points in $P$ are colinear.
A \emph{crossing-free geometric graph} on $P$ is a graph induced by the set of segments such that their endpoints are in $P$, and every two of them do not share their internal points.
In this paper, we are interested in \emph{connected} crossing-free geometric graphs, especially, \emph{crossing-free spanning trees} and \emph{crossing-free spanning cycles}.\footnote{Crossing-free spanning cycles are also called Hamiltonian cycles, spanning cycles, and planar traveling salesman tours.}
\cref{fig:cfg} shows examples of these geometric graphs.

\begin{figure}[t]
    \centering
    \begin{subfigure}{.49\linewidth}
        \centering
        \includegraphics[scale=0.4]{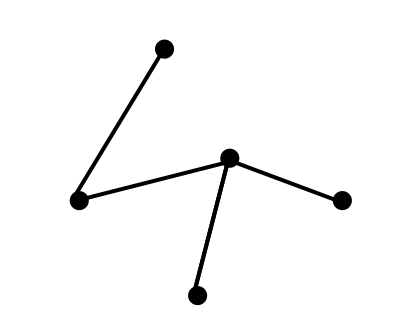}
        \caption{A crossing-free spanning tree.}
        \label{fig:tree}
    \end{subfigure}
    \begin{subfigure}{.49\linewidth}
        \centering
        \includegraphics[scale=0.4]{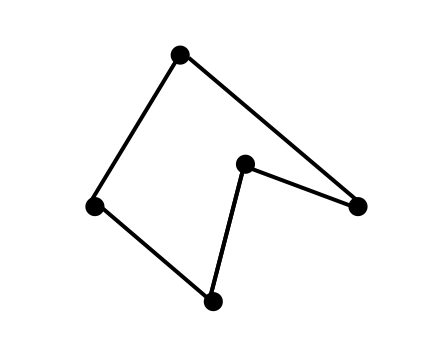}
        \caption{A crossing-free spanning cycle.}
        \label{fig:cycle}
    \end{subfigure}
    \caption{Examples of connected crossing-free geometric graphs.}
    \label{fig:cfg}
\end{figure}

We define $\mathrm{st}(P)$ and $\mathrm{sc}(P)$ as the numbers of crossing-free spanning trees and crossing-free spanning cycles on $P$, respectively.
One of the main research topics is to investigate the upper and lower bounds to $\mathrm{st}(P)$ and $\mathrm{sc}(P)$.
\cref{tab:bounds} summarizes the current best bounds.
For example, the top-left entry says that $\mathrm{st}(P) = \ostar{141.07^n}$ holds for all sets $P$ with $n$ points, see~\cite{HSSSTW13,SS11} (the $*$ indicates that any subexponential factors are ignored).
The up-to-date list of bounds for several crossing-free geometric graphs is available in~\cite{shefferWeb}.


Counting of crossing-free geometric graphs is also studied from an algorithmic point of view.
Although there are several problem-specific algorithms~\cite{AAHV07,AF96,KT09,Soh99,YHOUY19,ZSSM96}, there also exist general frameworks that can be applied to various crossing-free geometric graphs.
The first one is based on onion layer structures~\cite{ABCR15}, which runs in $n^{\bigO{k}}$ time, where $k$ is the number of onion layers.
Wettstein~\cite{Wet17} proposed a framework of algorithms which run in $\ostar{c^n}$ time for some constant $c$.
The framework allows efficient counting, enumeration, uniform random sampling, and optimization.
Currently, the fastest counting algorithm is presented by Marx and Miltzow~\cite{MM16} and runs in $n^{\mathrm{O}(\sqrt{n})}$ time.
However, it is not explicitly shown that their algorithm leads to efficient enumeration, uniform random sampling, or optimization.
In addition, their full paper consists of 47 pages with elaborate analysis~\cite{MM16full}.
In this paper, we are interested in designing simple and fast algorithms that can be applied to several purposes, not only counting, and thus we focus on Wettstein's framework.




An overview of Wettstein's framework is as follows.
First, we \emph{compile} crossing-free geometric graphs into a directed acyclic graph (DAG).
Second, using the DAG, we can efficiently perform counting, enumeration, uniform random sampling, and optimization~\cite{AS13}.
One of the remarkable contributions of Wettstein is to allow dealing with geometric graphs with ``connectivity'',
since it is known to be difficult to efficiently represent geometric graphs with such \emph{global} property.
To achieve this, Wettstein proposed specialized techniques for spanning trees and spanning cycles and invented algorithms run in $\ostar{7.044^n}$ time and $\ostar{5.619^n}$ time, respectively.
Since it is known that there exists a point set with $\Omega^{*}(12.52^n)$ crossing-free spanning trees~\cite{HM15}, for such cases, the algorithm can compile all crossing-free spanning trees exponentially faster than explicitly enumerating them.
However, it is unclear that the compilation algorithm is \emph{always} exponentially faster than explicit enumeration because the lower bound for any point set $P$ is $\Omega^{*}(6.75^n)$~\cite{FN99}.
It was left as future work in Wettstein's paper whether or not we can reduce the base of the time complexity of the compilation algorithm to less than $6.75$.
For crossing-free spanning cycles, we cannot hope for such an exponential speed-up by compilation because $\mathrm{sc}(P) = 1$ holds for a set $P$ of $n$ points in convex position.
Although a point set with $\Omega^{*}(4.64^n)$ crossing-free spanning cycles is known~\cite{GNT00}, the base of the number is less than $5.619$, which is the base of the running time of Wettstein's compilation algorithm for crossing-free spanning cycles.
Therefore,
%
%
we have the following natural question:
does there exist a point set where the algorithm can compile all crossing-free spanning cycles exponentially faster than explicit enumeration?

\begin{table}[t]
    \centering
    \caption{Bounds for the numbers of connected crossing-free geometric graphs.}
    \begin{tabular}{lrrrr}
         & \multicolumn{2}{c}{$\mathrm{st}(P)$} & \multicolumn{2}{c}{$\mathrm{sc}(P)$} \\ \hline
        $\forall P: \ostar{c^n}$ & 141.07 & \cite{HSSSTW13,SS11} & 54.55 & \cite{SSW13} \\
        $\exists P: \Omega^{*}(c^n)$ & 12.52 & \cite{HM15} & 4.64 & \cite{GNT00} \\
        $\forall P: \Omega^{*}(c^n)$ & 6.75 & \cite{FN99} & 1.00 &
    \end{tabular}
    \label{tab:bounds}
\end{table}

%

%
%
Our contribution includes answers to the above two open questions.
Moreover, we show that our technique can be applied for solving optimization problems.
Wettstein's framework can be applied to various geometric objects. In this paper, we first focus on refining the framework to answer the questions, and next show that the framework can be used for optimizations.
Our compilation algorithms are based on Wettstein's framework, and the constructed DAG by our algorithm can be used for efficient counting, enumeration, random sampling, and optimization.
Now, we describe the detail of our contributions below.
First, we propose a technique to deal with the connectivity constraint more simply and efficiently.
It makes the design and analysis of algorithms easier, and yields improved time complexity for compilation algorithms.
Our algorithm can compile all crossing-free spanning trees in $\ostar{6^n}$ time and all crossing-free spanning cycles in $\ostar{4^n}$ time.
Since $\mathrm{st}(P) = \Omega^{*}(6.75^n)$ holds for any point set $P$~\cite{FN99}, our compilation algorithm for crossing-free spanning trees is \emph{always} exponentially faster than explicit enumeration.
For crossing-free spanning cycles, recall that we cannot hope for such an exponential speed-up because there exists a point set $P$ with $\mathrm{sc}(P) = 1$.
However, since there exists a point set $P$ with $\mathrm{st}(P) = \Omega^{*}(4.64^n)$~\cite{GNT00}, for such cases, our compilation algorithm for crossing-free spanning cycles is guaranteed to run exponentially faster than explicit enumeration.

Next, we propose an algorithm to optimize the area surrounded by spanning cycles using a DAG.
To achieve this, we modify the DAG so that it has additional information.
Our algorithm runs in $\ostar{4.829^n}$ time to find an area-minimized (or maximized) spanning cycle of a given point set.
Although the problem was shown to be NP-complete in 2000~\cite{Fek00}, as far as we know, there were no known algorithms faster than the obvious $\ostar{n!}$ time algorithm for 20 years.
To the best of our knowledge, our algorithm is the first such one.

In the following sections, we show the proofs of the lemmas and theorems marked with * in the appendix.

\section{Overview of Wettstein's framework}
In this section, we review Wettstein's framework.
Let $P$ be a set of $n$ points in the plane in general position, that is, no three points in $P$ are colinear.
Let $\sgm$ be the set of segments (line segments) whose endpoints are in $P$.
We assume that no two points have the same $x$-coordinate.
With this assumption, the points can be uniquely ordered as $p_1, \dots, p_n$ from left to right.
If $i \le j$ (or $i < j$), we write $p_i \preceq p_j$ (or $p_i \prec p_j$).
Two different segments $s_1$ and $s_2$ are \emph{non-crossing} if they do not share their internal points.
The set $C$ ($\subseteq \sgm$), called a \emph{combination} of $\sgm$, is \emph{crossing-free} if every two different segments in $C$ are non-crossing.

The basic idea of the framework is to represent a geometric graph as a combination of \emph{units}.
As units, we intensively consider segments, although Wettstein considered several units (e.g., triangles for triangulations).
Both crossing-free spanning trees and crossing-free spanning cycles can be expressed by the
%
%
sets of their $n - 1$ and $n$ segments, respectively.

To represent a set of geometric graphs, we define a special DAG.
\begin{definition}\label{def:combination_graph}
    A \emph{combination graph} is a directed and acyclic multigraph $\Gamma$ with two distinguished vertices $\bot$ and $\top$, called the \emph{source} and \emph{sink} of $\Gamma$. All edges in $\Gamma$, except for those ending in $\top$, are labeled with a segment in $\sgm$.
    Moreover, the sink $\top$ has no outgoing edges.
    The \emph{size} $\size{\Gamma}$ of $\Gamma$ is the number of vertices and edges in $\Gamma$.
\end{definition}

In a combination graph, there is a one-to-one correspondence between a $\bot$-$\top$ path with a combination of segments.
In other words, a $\bot$-$\top$ path represents a combination of segments that is the set of labels of edges appearing in the path.
Therefore, using a combination graph, we can represent a set of geometric graphs.

\cref{fig:combination_graph} shows an example.
\cref{fig:cycles} shows the set of three crossing-free spanning cycles on the same point set.
\cref{fig:dag} is a combination graph representing the set of crossing-free spanning cycles.
In the figure, $s_{ij}$ denote the segment whose endpoints are $p_i$ and $p_j$.
Each alphabet in a circle is the name of the vertex.
There are three $\bot$-$\top$ paths: $\bot$-A-C-F-H-$\top$, $\bot$-B-D-F-H-$\top$, and $\bot$-B-E-G-I-$\top$.
The paths correspond to the crossing-free spanning cycles in \cref{fig:cycles} from left to right.

\begin{figure}[t]
    \centering
    \begin{subfigure}{.49\linewidth}
        \centering
        \includegraphics[scale=0.4]{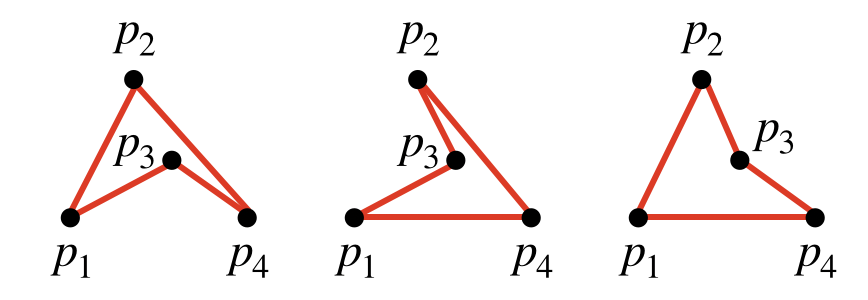}
        \caption{Crossing-free spanning cycles on the same point set.}
        \label{fig:cycles}
    \end{subfigure}
    \begin{subfigure}{.49\linewidth}
        \centering
        \includegraphics[scale=0.4]{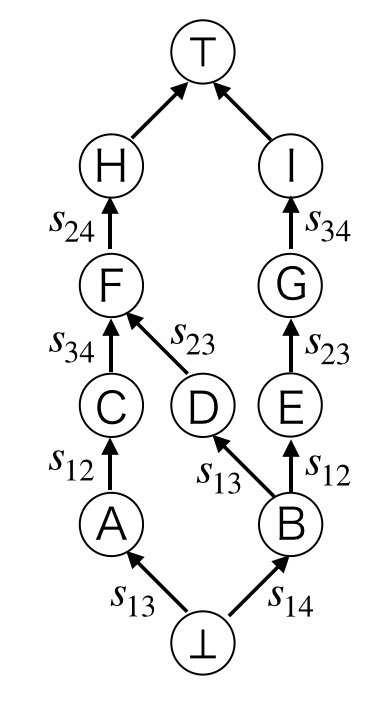}
        \caption{A combination graph representing the set of crossing-free spanning cycles.}
        \label{fig:dag}
    \end{subfigure}
    \caption{Examples of crossing-free spanning cycles and a corresponding combination graph.}
    \label{fig:combination_graph}
\end{figure}

Once we compile geometric graphs into a combination graph $\Gamma$, we can use $\Gamma$ for efficient counting, enumeration, random sampling, and optimization of ``decomposable'' function~\cite{AS13}.
One example of ``decomposable'' functions is the sum of lengths of segments in a combination.
More generally, we can optimize a linear function of $\sgm$.
Given a cost function $\func{c}{\sgm}{\mathbb{R}}$, we call a function $\func{f}{2^{\sgm}}{\mathbb{R}}$ a \emph{linear function} if $f$ is in the form $f(C) = \sum_{s \in C} c(s)$ for $C \subseteq \sgm$.
We summarize the uses of $\Gamma$ in the next lemma.
In the following lemma, \emph{solutions} mean the geometric graphs represented by a combination graph.
In fact, our time bound for random sampling improves the previous time bound appeared in \cite{AS13}.
We show details in \cref{app:query}.

\begin{lemma}[*]\label{lem:query}
    Let $\Gamma$ be a combination graph
    (whose edge labels are in $\mathcal{S}_P$) and $h$ be the \emph{height} of $\Gamma$, that is, the maximum number of edges contained in $\bot$-$\top$ paths. Then, we can
    \begin{itemize}
        \item count the number of solutions in $\bigO{\size{\Gamma}}$ time,
        \item enumerate solutions in $\bigO{h}$ time per solutions,
        \item randomly sample a solution in $\bigO{h \log n}$ time\footnote{Using a technique in \cite{AS13} yields time bound $\bigO{hn^2}$.
        However, we can reduce the time to $\bigO{h \log n}$. See \cref{app:query} for details.}, and
        \item find a solution minimizing (or maximizing) a given linear function of $\sgm$ in $\bigO{\size{\Gamma}}$ time.
    \end{itemize}
\end{lemma}

From now on, we describe how to construct a combination graph efficiently.
For a segment $s \in \sgm$, $\pts{s}$ denotes the set of endpoints of $s$.
We define $\lft{s}$ and $\rgt{s}$ as the left and right endpoint of $s$, respectively.
In other words, if $\pts{s} = \{p_i, p_j\}$ and $i < j$, then $\lft{s} = p_i$ and $\rgt{s} = p_j$.
For two segments $s_1$ and $s_2$, if $\rgt{s_1} \preceq \lft{s_2}$, then we write $s_1 \preceq s_2$.
For each $s \in \sgm$, we define $\low{s} \subseteq P$ and $\upp{s} \subseteq P$, the \emph{lower} and \emph{upper shadow} of $s$, respectively.
The set $\low{s}$ contains all points in $P$ from which a vertical ray shooting upwards intersects the relative interior of $s$.
The set $\upp{s}$ is defined analogously.
Whenever we have $\pts{s_1} \cap \low{s_2} \neq \emptyset$ or $\upp{s_1} \cap \pts{s_2} \neq \emptyset$ for any $s_1, s_2 \in \sgm$, then we say that $s_2$ \emph{depends on} $s_1$ and we write $s_1 \sqsubset s_2$.
For $C \subseteq \sgm$, $\pts{C}$ and $\low{C}$ denote the sets $\bigcup_{s \in C} \pts{s}$ and $\bigcup_{s \in C} \low{s}$, respectively.

Assume $C \subseteq \sgm$ holds.
Then, a segment $s \in C$ is
\emph{extreme} (\emph{in $C$}) if $s \not\depend s'$ holds for all $s' \in C \setminus \{s\}$.
If it exists, the \emph{right-most extreme element} in $C$ is the unique extreme element $s$ in $C$ such that $s' \preceq s$ for all extreme elements $s' \in C \setminus \{s\}$.

\begin{definition}\label{def:serializable}
    Let $\mf{C}$ be a set of combinations of $\sgm$.
    We call $\mf{C}$ \emph{serializable} if
    %
    %
    $\mf{C}$ is non-empty and if every non-empty $C \in \mf{C}$ contains a right-most extreme element,
    %
    %
    denoted by $\mathrm{rex}(C)$, and $C \setminus \{\mathrm{rex}(C)\}$ is an element of $\mf{C}$.
\end{definition}

Let $\mf{C}$ be a serializable set of combinations of $\sgm$.
For $C, C' \in \mf{C}$ and $s \in \sgm$, we write $C \xrightarrow{s} C'$ if $C = C' \setminus \{s\}$ and $s = \mathrm{rex}(C')$ hold.
Observe that $\mf{C}$ naturally induces a DAG,
%
%
which is almost a tree, as follows.
The graph has the vertex set $\mf{C}$ and the directed edges with labels from $\sgm$.
Whenever $C \xrightarrow{s} C'$ holds, we add an edge from vertex $C$ to vertex $C'$ with label $s$.
A combination graph representing an arbitrary subset of $\mf{C}$ is obtained by defining $\bot := \emptyset$ and by adding appropriate unlabeled edges pointing at an additional vertex $\top$.
However, the resulting combination graph is useless because its size is $\Theta(\size{\mf{C}})$.
To make the combination graph smaller, we define an equivalence relation among $\mf{C}$ and merge equivalent combinations.

\begin{definition}
    Let $\mf{C}$ be a serializable set of combinations of $\sgm$.
    An equivalence relation $\sim$ on $\mf{C}$ is
    %
    %
    \emph{coherent} if, for any $C_1, C_2 \in \mf{C}$ with $C_1 \sim C_2$, $C_1 \xrightarrow{s} C'_1$ impiles that $C_2 \xrightarrow{s} C'_2$ for some $C'_1 \sim C'_2$.
    In addition, if $C \not\sim C'$ holds for any $C, C' \in \mf{C}$ and $s \in \sgm$ satisfying
    $C \xrightarrow{s} C'$, we say that $\mf{C}$ is \emph{progressive on} $\sim$.
\end{definition}

For any $C \in \mf{C}$, we define the equivalence class $[C] := \inset{C' \in \mf{C}}{C' \sim C}$, where the relation $\sim$ will be obvious from the context.
We also define the set $(\mf{C}/{\sim}) := \inset{[C]}{C \in \mf{C}}$ of all equivalence classes.

If an equivalence relation $\sim$ on $\mf{C}$ is coherent, we can safely merge two vertices $C_1, C_2 \in \mf{C}$ such that $C_1 \sim C_2$.
In addition, progressiveness requires there are no loops in a combination graph.
Therefore, when $\mf{C}$ is a serializable set of combinations of $\sgm$ and $\sim$ is a coherent equivalence relation on $\mf{C}$ such that $\mf{C}$ is progressive on $\sim$, by merging equivalent vertices with respect to $\sim$, we obtain a DAG whose vertices correspond to equivalence classes.
For any subset $\mf{T}$ of equivalence classes $(\mf{C} / {\sim})$, we obtain a combination graph $\Gamma$ representing $\mf{T}$ by adding unlabeled edges from every vertex $[C] \in \mf{T}$ to $\top$.
The number of vertices in $\Gamma$ is $\size{(\mf{C} / {\sim})}$ and each vertex has at most $\size{\sgm} = \bigO{n^2}$ edges.
In summary, the following lemma holds.

\begin{lemma}[Lemma~2 in~\cite{Wet17}]\label{lem:size}
    Let $\mf{C}$ be a serializable set of combinations of $\sgm$, $\sim$ be a coherent equivalence relation such that $\mf{C}$ is progressive on $\sim$, and $\mf{T}$ be a subset of $(\mf{C}/{\sim})$.
    Then, there exists a combination graph $\Gamma$ with size $\bigO{\size{(\mf{C}/{\sim})} \cdot n^2}$ that represents $\bigcup_{[C] \in \mf{T}} [C]$.
\end{lemma}

To obtain a time bound to construct $\Gamma$, we add another factor to check a given combination is in $\mf{C}$.
As we will see in the later sections, it can be done in $\bigO{n}$ time for all problems discussed in this paper.

\section{Algorithms for connected crossing-free geometric graphs}\label{sec:connected}
\subsection{Crossing-free spanning trees}\label{sec:tree}
In this subsection, we propose an algorithm to compile crossing-free spanning trees.
By definition, $C \subseteq \sgm$ is a crossing-free spanning tree if and only if 1) $C$ is crossing-free, 2) $C$ is cycle-free, and 3) all the points are connected in $C$.
To ensure the first condition, we use Wettstein's algorithm to compile all crossing-free geometric graphs.
As for the second and the third condition, it suffices to maintain the connectivity of points in $C$.
To deal with the connectivity efficiently, we propose a new simple and efficient technique, which leads to an improved complexity.

First, we introduce Wettstein's algorithm to compile all crossing-free geometric graphs.
In the following, $\cfg$ denotes the set of crossing-free combinations of $\sgm$.
Note that, if $C \in \cfg$ is non-empty, then any subset of $C$ is in $\cfg$.
This property does not hold for crossing-free spanning trees and crossing-free spanning cycles.

Let $C \in \cfg$.
As in \cite{Wet17}, we partition $P$ into three sets $W(C)$, $G(C)$, and $B(C)$. (Each symbol stands for white, gray, and black.)
The sets are defined by $B(C) := \low{C}$, $G(C) := \pts{C} \setminus \low{C}$, and $W(C) := P \setminus \pts{C}$.
If there is no ambiguity, we omit $C$ and denote $W$, $G$, and $B$.
Note that $G$ is non-empty if $C$ is non-empty.
One point in $G$ is marked as $m(C)$ such that $m(C)$ is the left point of $\rex{C}$.
If $C = \emptyset$, we set $m(\emptyset) = \mathtt{nil}$.
We define $\tau(C) := (W, G, B, m)$ and the equivalence relation $\sim_{\tau}$ on $\cfg$ such that $C_1 \sim_{\tau} C_2$ if and only if $\tau(C_1) = \tau(C_2)$.
\begin{lemma}[Lemma~4 in \cite{Wet17}]\label{lem:cf_serialize}
    For any point set $P$, the set $\cfg$ is serializable.
\end{lemma}
\begin{lemma}[Lemma~3 in \cite{Wet17}]\label{lem:tau}
    The equivalence relation $\sim_{\tau}$ on $\cfg$ is coherent.
    In addition, $\cfg$ is progressive on $\sim_{\tau}$.
\end{lemma}

From now on, we propose our technique to deal with connectivity for crossing-free spanning trees.
To do this, we focus on the property of ``prefixes'' of crossing-free spanning trees.

Let $C, C' \subseteq \sgm$.
We call $C$ a \emph{prefix} of $C'$ if there exists a sequence of segments $s_1, \dots, s_k$ such that $C \xrightarrow{s_1} C_1, \dots, C_{k-1} \xrightarrow{s_k} C'$.
When $C$ is a prefix of $C'$, we say that $C'$ \emph{extends} $C$.
Let $U \subseteq P$ be a connected component in $C$.
We call $U$ a \emph{hidden component} if, for every point $p$ in $U$, there exists a segment $s \in C$ such that $p \in \low{s}$.
Intuitively, such a component is invisible from above because of other segments.
\begin{lemma}\label{lem:prestr}
    Let $C \subseteq \sgm$ be a set of crossing-free segments.
    If $C$ is a prefix of a crossing-free spanning tree $C^* \in \str$, then all of the following hold:
    \begin{enumerate}
        \item[A1.] $C$ is cycle-free, and
        \item[A2.] there are no hidden components in $C$.
    \end{enumerate}
\end{lemma}
\begin{proof}
    A1 is obviously necessary.
    Assume that $C$ violates A2.
    Now, there exists a hidden component $U$ in $C$.
    In any $C'$ extending $C$, $r := \rex{C'}$ is not incident to any point in $U$.
    If $r$ is incident to a point $p$ in $U$, since there exists a segment $s$ in $C'$ such that $p \in \low{s}$, we have $r \depend s$, contradicting that $r$ is the right-most extreme element in $C'$.
    It follows that, in any $C'$ extending $C$, there are at least two connected components in $C'$: $U$ and the one containing $r$.
    This means that $C$ cannot be extended to any crossing-free spanning tree, contradicting that $C$ is a prefix of a crossing-free spanning tree.
\end{proof}

We define $\str$ as the set of crossing-free spanning trees on $P$ and $\prestr$ as the set of combinations of $\sgm$ satisfying both conditions A1 and A2 in \cref{lem:prestr}.
In other words, $\prestr$ is the superset of real prefixes of crossing-free spanning trees on $P$.
Especially, $\prestr$ properly contains $\str$, and thus we only have to consider $\prestr$ to obtain $\str$.
Note that, for every non-empty $C \in \prestr$, removing $\rex{C}$ from $C$ does not violate any conditions in \cref{lem:prestr}.
Therefore, we obtain the following lemma.
\begin{lemma}\label{lem:prestr_serializable}
    For any point set $P$, the set $\prestr$ is serializable.
\end{lemma}

Now, we propose our technique to deal with connectivity.
For $C \in \prestr$, we define the partition $\Pi(C)$ of $G(C)$ such that two points $x, y \in G(C)$ are connected in $C$ if and only if they are in the same set in $\Pi(C)$.
If there is no ambiguity, we omit $C$ from $\Pi(C)$ and denote $\Pi$.
Finally, we define $\phi(C) = (W, G, B, m, \Pi)$ and the equivalence relation $\sim_{\phi}$ on $\prestr$ such that $C_1 \sim_{\phi} C_2$ if and only if $\phi(C_1) = \phi(C_2)$.
\cref{fig:prestr} shows two equivalent elements of $\prestr$.
In the figure, white, gray, and black points are in $W$, $G$, and $B$, respectively.
A point with a bold circle is $m$.
The ellipses indicate the partition $\Pi$ of $G$.

\begin{figure}[t]
    \centering
    \begin{subfigure}{.30\linewidth}
        \centering
        \includegraphics[scale=0.5]{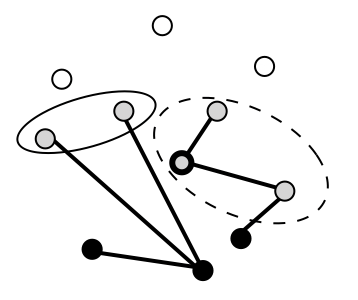}
    \end{subfigure}
    \begin{subfigure}{.30\linewidth}
        \centering
        \includegraphics[scale=0.5]{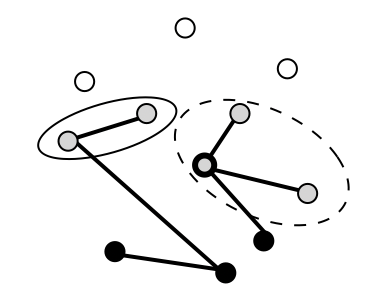}
    \end{subfigure}
    \caption{Two equivalent elements of $\prestr$.}
    \label{fig:prestr}
\end{figure}

\begin{lemma}[*]\label{lem:phi_coherent}
    The equivalence relation $\sim_{\phi}$ on $\prestr$ is coherent.
    In addition, $\prestr$ is progressive on $\sim_{\phi}$.
\end{lemma}

To obtain a bound on the size of a combination graph representing $\str$, we analyze the number of equivalence classes, that is, $\size{(\prestr / {\sim_{\phi}})}$.
Let us begin with a rough estimation.
The number of possible $(W, G, B)$'s is $\bigO{3^n}$.
The number of possible $m$'s is $\bigO{n}$.
Since $\Pi$ is a partition of at most $n$ points, the number of possible $\Pi$'s is at most the $n$-th Bell number, which is $\bigO{n^n}$.
It follows that $\size{(\prestr / {\sim})} = \bigO{(3n)^n n}$.
This bound is not in the form $\ostar{c^n}$ for some constant $c$.
From now on, we show
%
%
a considerably smaller
estimation on the number of equivalence classes.

Fortunately, we have the following observations.
A partition of $\set{1,\dots,N}$ is \emph{non-crossing}\footnote{The word ``non-crossing'' for a partition is independent from the word ``crossing-free'' for a set of segments.}~\cite{Sim00} if, for every four elements $1 \le a < b < c < d \le N$, $a, c$ are in the same set and $b, d$ are in the same set, then the two sets coincide.

\begin{lemma}\label{lem:non-crossing}
     For any $C \in \prestr$, let us order the points in $G$ as $p'_1, \dots, p'_{|G|}$ from left to right.
     Then, $\Pi$ is a non-crossing partition of $\set{p'_1, \dots, p'_{|G|}}$.
\end{lemma}
\begin{proof}
    For $1 \le i < j < k < l \le |G|$, assume that $p'_i$ and $p'_k$ are in $U_1 \in \Pi$ and $p'_j$ and $p'_l$ are in $U_2 \in \Pi$.
    Now, we assume that $U_1 \neq U_2$, which leads to a contradiction.
    By A1 ($C$ is cycle-free), there is the unique path $S_1$ from $p'_i$ to $p'_k$ in $C$.
    Since $p'_j$ is in $G = \pts{C} \setminus \low{C}$, the path $S_1$ passes under $p'_j$.
    Likewise, there is the unique path $S_2$ in $C$ from $p'_j$ to $p'_l$ and it passes under $p'_k$.
    Since $U_1 \neq U_2$, the paths $S_1$ and $S_2$ do not share their vertices.
    This means that there exists a pair of segments $s_1 \in S_1$ and $s_2 \in S_2$ such that they are crossing, which contradicts that $C$ is crossing-free.
\end{proof}

It is known that the number of non-crossing partitions of $N$ elements is the $N$-th Catalan number $\catl{N}$~\cite{Sim00}, which is at most $4^N$~\cite[p.~450]{Knu11}.
This is much smaller than the number of general partitions, the $N$-th Bell number, which is $\bigO{N^N}$.
Using these facts, we can improve the previous rough estimation on the number of the equivalence classes.
Since the number of $(W, G, B)$ is $\bigO{3^n}$, $m$ is $\bigO{n}$, and $\Pi$ is $\bigO{4^n}$, it follows that $\size{(\prestr / {\sim_{\phi}})} = \bigO{12^n}$.
Now, we have obtained a bound with the form $\ostar{c^n}$ for a constant $c$.
However, this estimation is still rough because $\Pi$ does not always contain $n$ points and contains only the points in $G$.
The following lemma shows a substantially smaller estimation on the number of the equivalence classes.

\begin{lemma}\label{lem:phi_size}
    $\size{(\prestr / {\sim_{\phi}})} = \bigO{6^n n}$.
\end{lemma}
\begin{proof}
    For every partition $(W, G, B)$ of $P$, $\Pi$ is a non-crossing partition of $|G|$ elements.
    Therefore, the number of $\Pi$ is at most $\catl{|G|} \le 4^{|G|}$ when we fix $(W, G, B)$.
    Let $i$, $j$, and $k$ be the sizes of $W$, $G$, and $B$, respectively.
    The number of $(W, G, B)$ such that $|W| = i$, $|G| = j$, and $|B| = k$, is $\displaystyle\frac{n!}{i!\,j!\,k!}$.
    Using these facts and the multinomial theorem, the number of $(W, G, B, \Pi)$ is at most
    \begin{align*}
        \sum_{\substack{W, G, B \subseteq P,\\ W \cup G \cup B = P,\\ W \cap G = G \cap B = W \cap B = \emptyset}} \catl{|G|}
        \quad &= \sum_{i + j + k = n} \frac{n!}{i!\,j!\,k!} \cdot \catl{j} \\
        \quad &\le \sum_{i + j + k = n} \frac{n!}{i!\,j!\,k!} \cdot 1^i \cdot 4^j \cdot 1^k \\
        \quad &= (1 + 4 + 1)^n = 6^n.
    \end{align*}
    Since the number of $m$ is at most $n$, we obtain $\size{(\prestr / {\sim_{\phi}})} = \bigO{6^n n}$.
\end{proof}

From \cref{lem:size,lem:prestr_serializable,lem:phi_coherent,lem:phi_size}, we obtain the bound on the size of a combination graph representing $\prestr$.
To show the time complexity to construct the combination graph, it suffices to add another factor $n$ to the size of the combination graph because we can check the membership in $\prestr$ in $\bigO{n}$ time, as shown in \cref{app:spanning_tree}.

\begin{theorem}[*]\label{theo:spanning_tree}
    Let $P$ be a set of $n$ points in the plane in general position. Then, there exists a combination graph of size $\bigO{6^n n^3}$ that represents $\str$.
    We can construct it in $\bigO{6^n n^4}$ time.
\end{theorem}

\subsection{Crossing-free spanning cycles}\label{sec:cycle}
In this subsection, we propose an algorithm to compile crossing-free spanning cycles.
By definition, $C \subseteq \sgm$ is a crossing-free spanning cycle if and only if 1) $C$ is crossing-free, 2) all the points have degree 2 in $C$, and 3) all the points are connected in $C$.
To deal with the second condition, we modify the definition of $W$, $G$, and $B$ in \cref{sec:tree} so that they partition $P$ into the sets of points that have the same degree.
In fact, using these modified $W$, $G$, and $B$, we can check the first condition.
To deal with the third condition, we propose a specialized technique to deal with connectivity for crossing-free spanning cycles, which leads to a better complexity than the algorithm for crossing-free spanning trees.



As we have done in \cref{sec:tree}, we focus on the property of the prefixes of crossing-free spanning cycles.
In the following, the \emph{degree} of a point $p$ in $C$ is the number of segments incident to $p$, denoted by $\degr{p}{C}$.
We call $U \subsetneq P$ is an \emph{isolated cycle} in $C$ if $U$ is a connected component in $C$ and all the points in $U$ have degree 2.
Note that an isolated cycle is not necessarily a hidden component.
\begin{lemma}\label{lem:prescy}
    Let $C \subseteq \sgm$ be a set of crossing-free segments.
    If $C$ is a prefix of a crossing-free spanning cycle $C^* \in \scy$, then all of the following hold:
    \begin{enumerate}
        \item[B1.] $\degr{p}{C} \le 2$ for every point $p$,
        \item[B2.] for every point $p$ of $\degr{p}{C} < 2$, $C$ has no segment $s$ such that $p \in \low{s}$, and
        \item[B3.] there are no isolated cycles in $C$.
    \end{enumerate}
\end{lemma}
\begin{proof}
    Since $\degr{p}{C^*} = 2$ for every point $p$ and a crossing-free spanning cycle $C^*$, B1 is necessary.
    Assume that $C$ does not satisfy B2.
    Then, there is a point $p$ such that $\degr{p}{C} < 2$ and a segment $s \in C$ such that $p \in \low{s}$.
    For any $C'$ extending $C$, $\rex{C'}$ is not incident to $p$ because, if so, $\rex{C'} \depend s$, which is a contradiction.
    It means that $\degr{p}{C'} = \degr{p}{C} < 2$ for any $C'$ extending $C$, which contradicts that $C$ is a prefix of a crossing-free spanning cycle.

    Assume that $C$ does not satisfy B3.
    Then, there exists an isolated cycle $U \subsetneq P$ in $C$.
    If there exists a crossing-free spanning cycle $C^*$ extending $C$, then $C^*$ contains at least one segment in $C^* \setminus C$ incident to a point $p \in U$.
    However, this means that $\degr{p}{C^*} > \degr{p}{C} = 2$, contradicting that $C^*$ is in $\scy$.
\end{proof}

We define $\scy$ as the set of all crossing-free spanning cycles and $\prescy$ as the set of combinations of $\sgm$ satisfying all the conditions in \cref{lem:prescy}.
Since $\scy \subseteq \prescy$ holds, to obtain $\scy$, it suffices to consider $\prescy$.
Note that all the conditions from B1 to B3 are maintained when removing $\rex{C}$ from any $C \in \prescy$, which shows the following lemma.
\begin{lemma}\label{lem:prescy_serializable}
    For any point set $P$, the set $\prescy$ is serializable.
\end{lemma}

\begin{figure}[t]
    \centering
    \begin{subfigure}{.30\linewidth}
        \centering
        \includegraphics[scale=0.5]{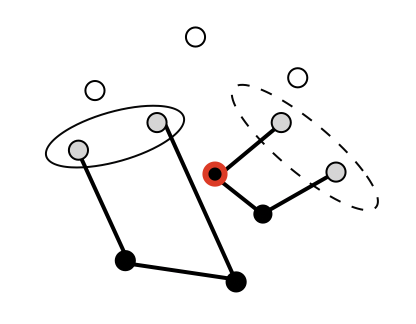}
    \end{subfigure}
    \begin{subfigure}{.30\linewidth}
        \centering
        \includegraphics[scale=0.5]{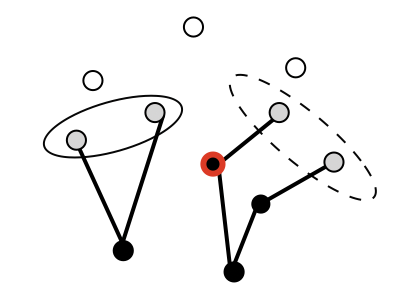}
    \end{subfigure}
    \caption{Two equivalent elements of $\prescy$.}
    \label{fig:prescy}
\end{figure}

We define an equivalence relation on $\prescy$.
For this purpose, we modify the definitions of $W, G$, and $B$ in \cref{sec:tree} so that they are the sets of points with degree 0, 1, and 2, respectively.
Note that all the points in $\low{C}$ are in $B$.
Moreover, we put a point $p$ in $\pts{C} \setminus \low{C}$ into $B$ if $p$ has degree 2.
$W$ is the same as \cref{sec:tree}: $W = P \setminus \pts{C}$.
As a result, $G$ contains an even number of points because the total degree must be even.
We define mark $m$ in the same way as \cref{sec:tree}.

By \cref{lem:prescy}, every $C \in \prescy$ is a disjoint union of paths such that every point in $G$ (or $B$) is an endpoint (or an internal point) of a path.
Therefore, to deal with connectivity, we define a matching $\mathcal{M}$ of points in $G$ such that, for every two different points $x, y \in G$, $x$ and $y$ are the two endpoints of a path in $C$ if and only if they are matched in $\mathcal{M}$.
In other words, $\mathcal{M}$ is a partition of set $G$ of gray points such that every set contains exactly two points.


We define $\psi(C) := (W, G, B, m, \mathcal{M})$ and the equivalence relation $\sim_{\psi}$ on $\prescy$ such that $C \sim_{\psi} C'$ if and only if $\psi(C) = \psi(C')$.
\cref{fig:prescy} shows two equivalent elements of $\prescy$.
In the figure, white, gray, and black points are the points in $W$, $G$, and $B$, respectively.
The point with a bold circle is $m$.
The ellipses indicate the matching $\mathcal{M}$ of $G$.

\begin{lemma}[*]\label{lem:psi_coherent}
    The equivalence relation $\sim_{\psi}$ on $\prescy$ is coherent.
    In addition, $\prescy$ is progressive on $\sim_{\psi}$.
\end{lemma}

The number of matchings of $N$ elements is $(N'-1)!! = (N'-1) \cdot (N'-3) \cdot \dots \cdot 1$ where $N'$ is the largest even number such that $N' \le N$.
Although this is slightly smaller than the number of general partitions of $N$ elements, it is still $\bigO{N^N}$.
This prevents us from obtaining the bound in the form $\ostar{c^n}$ for some constant $c$.
However, an appropriate analysis shows that there is a considerably smaller estimation on the number of possible $\mathcal{M}$'s.

By \cref{lem:non-crossing}, $\mathcal{M}$ is a non-crossing partition of $G$ whose every set contains exactly two elements.
Such a partition has a one-to-one correspondence with a balanced sequence of $|G|$ parentheses.
The correspondence is defined as follows.
Let $p'_1, \dots, p'_{|G|}$ are the points in $G$ ordered from left to right.
For two points $p'_i$ and $p'_j$ with $i < j$, if they are the two endpoints of the same path in $C$,
we put `(' and `)' in the $i$-th and $j$-th position of a sequence.
Since paths are pairwise non-crossing, we obtain a balanced sequence of parentheses in this way.
\cref{fig:parentheses} shows the correspondence between a matching of the endpoints of pairwise non-crossing paths and a balanced sequence of parentheses.
Using this one-to-one correspondence, we obtain the following lemma.

\begin{figure}[t]
    \centering
    \includegraphics[scale=0.5]{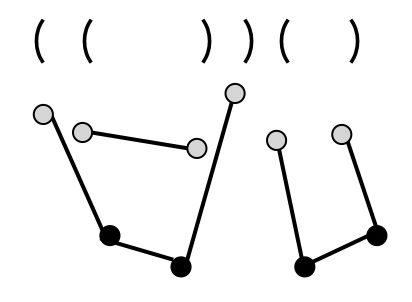}
    \caption{The correspondence between $\mathcal{M}$ and a balanced sequence of parentheses.}
    \label{fig:parentheses}
\end{figure}

\begin{lemma}\label{lem:psi_size}
    $\size{(\prescy / {\sim}_{\psi})} = \bigO{4^n n}$.
\end{lemma}
\begin{proof}
    From the one-to-one correspondence between the matching $\mathcal{M}$ and a balanced sequence of $|G|$ parentheses, the number of possible $\mathcal{M}$'s is $2^{|G|}$.
    Therefore, the number of $(W, G, B, \mathcal{M})$ is at most
    \begin{align}
        \sum_{\substack{W, G, B \subseteq P,\\ W \cup G \cup B = P,\\ W \cap G = G \cap B = W \cap B = \emptyset}} 2^{|G|}
        \quad = \sum_{i + j + k = n} \frac{n!}{i!j!k!} \cdot 1^i \cdot 2^j \cdot 1^k
        = (1 + 2 + 1)^n = 4^n.
    \end{align}
    Since the number of $m$ is at most $n$, it follows that $\size{(\mf{C}^\mathsf{sc} / {\sim}_{\psi})} = \bigO{4^n n}$.
\end{proof}

By \cref{lem:size,lem:prescy_serializable,lem:psi_coherent,lem:psi_size} and showing that the membership in $\prescy$ can be checked in $\bigO{n}$ time, we obtain the following
%
%
theorem.

\begin{theorem}[*]\label{theo:spanning_cycle}
    For any $P$, there exists a combination graph of size $\bigO{4^n n^3}$ that represents $\scy$.
    We can construct it in $\bigO{4^n n^4}$ time.
\end{theorem}

\section{Optimizing the area of crossing-free spanning cycles}\label{sec:area}
In this section, we propose an algorithm to optimize the area surrounded by crossing-free spanning cycles.

Given a combination graph $\Gamma$ whose edge labels are in $\sgm$, by \cref{lem:query}, one can optimize a linear function of $\sgm$ in $\bigO{\size{\Gamma}}$ time.
One example of such a function is the sum of lengths of segments.
Therefore, we can find a crossing-free spanning cycle with minimum (or maximum) length in $\bigO{4^n n^4}$ time by \cref{theo:spanning_cycle}.
At first glance, the area does not seem to be a linear function of $\sgm$.
However, using a well-known technique in computational geometry, we can express the area as a linear function of \emph{directed segments}.
Therefore, to optimize the area, we construct a combination graph whose edge labels are directed segments.



The following lemma is a well-known technique to calculate the area of a simple polygon, and thus a crossing-free spanning cycle.

\begin{lemma}[Green's theorem, or Exercise 33.1-8 in~\cite{CLRS09}]\label{lem:green}
    Let $S$ be a simple polygon with $N$ vertices ordered in counter-clockwise order as $(x_1, y_1), \dots, (x_N, y_N)$.
    For convenience, we define $(x_{N+1}, y_{N+1}) = (x_1, y_1)$.
    The area $A(S)$ of $S$ is
    \begin{equation}\label{eq:green}
        A(S) = \sum_{i = 1}^{N} \frac{(x_{i+1} + x_{i}) (y_{i+1} - y_{i})}{2}.
    \end{equation}
\end{lemma}
\cref{eq:green} calculates $A(S)$ as the sum of the signed area of the trapezoid defined by each edge.
The $i$-th edge of $S$ defines the trapezoid whose vertices are $(x_i, y_i)$, $(x_{i+1}, y_{i+1})$, $(0, y_{i+1})$, and $(0, y_{i})$.
The area of the trapezoid is positive if the vertices are ordered counter-clockwise and negative otherwise.

Let $\dsg$ be the set of directed segments whose endpoints are in $P$.
In other words, $\dsg := \bigcup_{\set{a, b} \in \sgm} \set{(a, b), (b, a)}$.
For a point $p \in P$, $(x_p, y_p)$ denotes its coordinates.
We call $\vec{C} \subseteq \dsg$ a \emph{counter-clockwise (crossing-free) spanning cycle} if $\vec{C}$ satisfies the following: 1) $(a, b) \in \vec{C}$ implies $(b, a) \notin \vec{C}$, 2) the set $\inset{\set{a, b}}{(a, b) \in \vec{C}}$ of segments is a crossing-free spanning cycle, and 3) the segments in $\vec{C}$ are directed in counter-clockwise order.
The \emph{area}, denoted by $A(\vec{C})$, of $\vec{C}$ means the enclosed area by the polygon defined by $\vec{C}$.
Then, \cref{eq:green} can be written as
\begin{equation}\label{eq:green_segments}
    A(\vec{C}) = \sum_{(a, b) \in \vec{C}} \frac{(x_{b} + x_{a}) (y_{b} - y_{a})}{2}.
\end{equation}
\cref{eq:green_segments} expresses the area of a  counter-clockwise spanning cycle as a linear function of directed segments.

On the basis of the above discussion, we construct a combination graph $\Gamma'$ representing $\dsc \subseteq 2^{\dsg}$, where $\dsc$ denote the set of counter-clockwise spanning cycles on $P$.
The edges of $\Gamma'$ are labeled by directed segments in $\dsg$.
We associate each directed segment $\vec{s} = (a, b) \in \dsg$ with weight $c(\vec{s}) := (x_{b} + x_{a}) (y_{b} - y_{a}) / 2$.
Then, the sum of the weights in a $\bot$-$\top$ path in $\Gamma'$ is the area of the counter-clockwise spanning cycle represented by the path.
Therefore, to find an area-minimized (or maximized) counter-clockwise spanning cycle, it suffices to find a $\bot$-$\top$ path with minimum (or maximum) weight in $\Gamma'$.
It can be found in $\bigO{\size{\Gamma'}}$ time by \cref{lem:query}.

To construct $\Gamma'$, we focus on the property of prefixes of a counter-clockwise spanning cycle.
For a directed segment $\vec{s} = (a, b) \in \dsg$, we call $a$ the \emph{head} and $b$ the \emph{tail} of $\vec{s}$.
The \emph{in-degree} (or \emph{out-degree}) of a point $p$ in $\vec{C} \subseteq \dsg$ is the number of directed segments whose tail (or head) is $p$, denoted by $\indegr{p}{\vec{C}}$ (or $\outdegr{p}{\vec{C}}$).
The following lemma is obtained in a similar way as \cref{lem:prescy}.
\begin{lemma}\label{lem:predsc}
    Let $\vec{C} \subseteq \dsg$ a set of crossing-free directed segments.
    If $\vec{C}$ is a prefix of a counter-clockwise spanning cycle, then all of the following hold:
    \begin{enumerate}
        \item[C1.] $\indegr{p}{\vec{C}} \le 1$ and $\outdegr{p}{\vec{C}}\le 1$ for every point $p$,
        \item[C2.] for every point $p$ such that $\indegr{p}{\vec{C}} = 0$ or $\outdegr{p}{\vec{C}} = 0$, $\vec{C}$ has no directed segment $\vec{s}$ such that $p \in \low{\vec{s}}$, and
        \item[C3.] there are no isolated cycles in $\vec{C}$.
    \end{enumerate}
\end{lemma}
In the following, $\predsc$ denotes the set of combinations of $\dsg$ that satisfy the conditions from C1 to C3.
Note that all the conditions from C1 to C3 are maintained when removing $\rex{\vec{C}}$ from any $\vec{C} \in \predsc$, which shows the following lemma.
\begin{lemma}\label{lem:predsc_serializable}
    For any point set $P$, the set $\predsc$ is serializable.
\end{lemma}

We define an equivalence relation on $\predsc$ as follows.
For $\vec{C} \in \predsc$, we define $W(\vec{C})$, $G(\vec{C})$, and $B(\vec{C})$ in almost the same way as \cref{sec:cycle}.
More precisely, $W(\vec{C})$ is the set of points that have both indegree and outdegree 0.
The set $G(\vec{C})$ consists of points that have indegree 1 and outdegree 0, or in-degree 0 and out-degree 1.
$B(\vec{C})$ is the set of points with indegree 1 and outdegree 1.
We define $m(\vec{C})$ in the same way as \cref{sec:cycle}.

By \cref{lem:predsc}, every $\vec{C} \in \predsc$ is a disjoint union of \emph{directed} paths.
Therefore, to deal with connectivity and the directions of the paths, we define $\vec{\mathcal{M}}(\vec{C})$ as the directed variant of $\mathcal{M}$ in \cref{sec:cycle}.
$\vec{\mathcal{M}}$ is the set of pairs $(a, b)$ of $a, b \in P$ such that 1) every $p \in P$ appears in exactly one pair in $\vec{\mathcal{M}}$ and, 2) for every two points $x, y \in G$, $x$ and $y$ are respectively the head and the tail of the same path if and only if $(x, y)$ is in $\vec{\mathcal{M}}$.

We define $\chi(\vec{C}) := (W, G, B, m, \vec{\mathcal{M}})$ and the equivalence relation $\vec{C} \sim_{\chi} \vec{C'}$ if and only if $\chi(\vec{C}) = \chi(\vec{C'})$.
We can prove the following lemma in almost the same way as \cref{lem:psi_coherent}.

\begin{lemma}
    The equivalence relation $\sim_{\chi}$ on $\predsc$ is coherent.
    In addition, $\predsc$ is progressive on $\sim_{\chi}$.
\end{lemma}

To enforce counter-clockwise order, we focus on the leftmost point $p_1$.
For a directed (clockwise or counter-clockwise) spanning cycle $\vec{C}$, let $u^\mathrm{in}$ be the unique segment of $\vec{C}$ whose tail is $p_1$ and $u^\mathrm{out}$ be the one whose head is $p_1$.
The segments in $\vec{C}$ are directed counter-clockwise if and only if $u^\mathrm{in}$ is above $u^\mathrm{out}$.
By the rule of extension, $u^\mathrm{out}$ is adopted before $u^\mathrm{in}$.
Therefore, $p_1$ never be the tail of some path.
This occurs if and only if $\vec{\mathcal{M}}$ contains a pair $(q, p_1)$ for some $q \neq p_1$.
We prune such cases.
It affects the time complexity of the algorithm only in a constant factor and does not enlarge the number of equivalence classes.

We analyze the overhead by maintaining the directions of the paths.
Since there exists at most $n/2$ paths, the bound $4^n \cdot 2^{n/2} = (4\sqrt{2})^n < 5.659^n$ is easy to obtain.
The following lemma shows that the bound can be further reduced.

\begin{lemma}
    $\size{(\vec{\mf{C}}^\mathsf{st} / {\sim}_{\chi})} = \bigO{c^n n}$ for some constant $c < 4.829$.
\end{lemma}
\begin{proof}
    When the number of points in $G$ is $j$, the number of paths is $j / 2$. (Note that $j$ is even.)
    Thus, the number of directions of the paths is $2^{j/2}$.
    Therefore, the number of $(W, G, B, \vec{\mathcal{M}})$ is at most
    \begin{align*}
        \sum_{\substack{W, G, B \subseteq P,\\ W \cup G \cup B = P,\\ W \cap G = G \cap B = W \cap B = \emptyset}} 2^{|G|} \cdot 2^{|G| / 2}
        &\quad = \sum_{i + j + k = n} \frac{n!}{i!j!k!} \cdot 1^i \cdot (2\sqrt{2})^j \cdot 1^k \\
        &\quad = (1 + 2\sqrt{2} + 1)^n = \bigO{c^n},
    \end{align*}
    where $c$ is a constant less than $4.829$.
    Since the number of $m$ is at most $n$, we obtain $\size{(\vec{\mf{C}}^\mathsf{st} / {\sim}_{\chi})} = \bigO{c^n n}$.
\end{proof}

Finally, we obtain the following theorem in the same way as \cref{lem:psi_coherent}.

\begin{theorem}\label{theo:directed_cycle}
    For any $P$, there exists a combination graph $\Gamma'$ of size $\bigO{c^n n^3}$ that represents $\dsc$ over $\dsg$, where $c$ is a constant less than $4.829$.
    We can construct it in $\bigO{c^n n^4}$ time.
    Within the same time bound, we can find an area-minimized (or maximized) crossing-free spanning cycle.
\end{theorem}

\section{Concluding remarks}
In this paper, we presented algorithms to compile and optimize connected crossing-free geometric graphs on
%
%
a set of $n$ points in the plane.
Our compilation algorithms run in $\ostar{6^n}$ time for crossing-free spanning trees and $\ostar{4^n}$ time for crossing-free spanning cycles.
In addition, we can find an area-minimized (or maximized) crossing-free spanning cycle in $\ostar{4.829^n}$ time.

One future direction is applying our technique to other geometric graphs with connectivity constraints.
Since our technique is simple, we believe that we can easily adapt it for other geometric graphs such as spanning forests, connected graphs (not necessarily spanning), spanning connected graphs (not necessarily cycle-free), and so on.

Another direction is improving our analysis of the complexity.
In this paper, we first defined an equivalence relation using the partition (or matching) of the points and then improved the bound using the multinomial theorem.
In contrast, Wettstein defined an equivalence relation using the coloring of points and then improved the bound using the fact that certain patterns cannot occur in the colorings because of geometric constraints.
Can we combine Wettstein's technique of analysis with our proof using the multinomial theorem?
Especially, for our algorithm to optimize the area of crossing-free spanning cycles, it is open whether there exists a point set where the algorithm runs exponentially faster than explicit enumeration because the current best lower bound to the maximum value of $\mathrm{sc}(P)$ is $\Omega^{*}(4.64^n)$~\cite{GNT00}.
Can we reduce the base of the time of our algorithm to less than 4.64?

\bibliography{reference}

\appendix

\section{Proof of \cref{lem:query}}\label{app:query}
For a vertex $\alpha$ of $\Gamma$, we define $\cnt{\alpha}$ as the number of $\alpha$-$\top$ paths in $\Gamma$.
We set $\cnt{\top} = 1$.
If $\alpha \neq \top$, then $\cnt{\alpha} = \sum_{(\alpha, \beta) \in E} \cnt{\beta}$, where $E$ is the edge set of $\Gamma$.
Using the formulas, we can recursively calculate $\cnt{\alpha}$ for every $\alpha$ in a reverse topological order.
Since there is the one-to-one correspondence between a solution and a $\bot$-$\top$ path, the number of solutions equals $\cnt{\bot}$.
It can be computed in $\bigO{\size{\Gamma}}$ time.

Enumeration can be done by depth-first search.
It takes $\mathrm{O}(h)$ time per solutions.

For random sampling, we use counting information.
First, we choose a random number $r$ between 1 and $\cnt{\top}$.
Then, we find the $r$-th $\bot$-$\top$ path in the following way.
We start from $\bot$.
Until we reach $\top$, we repeat the following.
When we are on a vertex $\alpha$,
we order the descendants of $\alpha$ in an arbitrary way: $\beta_1, \dots, \beta_k$.
Starting from $i = 1$, while $r > \cnt{\beta_i}$, we apply $r \gets r - \cnt{\beta_i}$ and $i \gets i + 1$.
We move to $\beta_j$, if $j$ is the first integer such that $r \le \cnt{\beta_j}$.
Until one reaches $\top$, there are at most $h$ vertices.
For each vertex, it takes $\bigO{n^2}$ time to select the appropriate descendant.
Thus, the total time is $\bigO{h n^2}$.
However, we can reduce the time to $\bigO{h \log n}$.
To achieve it, we calculate the cumulative sums for each vertex as a preprocessing.
We define $\mathrm{sum}(\alpha, i) := \sum_{j \le i} \cnt{\beta_j}$.
The number of the cumulative sums is the number of edges, so calculating all cumulative sums enlarges the complexity only by a constant factor.
Using the cumulative sums, we can find the appropriate descendant by binary search, which takes $\bigO{\log n^2} = \bigO{\log n}$ time.
Therefore, the total running time is $\bigO{h \log n}$.

Optimizing a linear function reduces to finding a shortest or longest $\bot$-$\top$ path in $\Gamma$.
This can be done in $\bigO{\size{\Gamma}}$ time in a similar way as counting.


\section{Proof of \cref{lem:phi_coherent}}
Let $C_1, C_2 \in \prestr$ be non-empty (otherwise, the proof is trivial) with $C_1 \sim_{\phi} C_2$ and assume that $C_1 \xrightarrow{s} C_1'$ holds for $C_1' \in \prestr$ and $s \in \sgm$.
Consider $C_2' := C_2 \cup \set{s}$.
We show $C_2' \in \prestr$, $C_1' \sim_{\phi} C_2'$, and $C_2 \xrightarrow{s} C_2'$, which implies the coherency of $\sim_{\phi}$.

Since $\sim_{\tau}$ is coherent by \cref{lem:tau}, $C_2'$ is crossing-free.
Therefore, what is left for us to show $C_2' \in \prestr$ is that $C_2'$ satisfies both A1 and A2 in \cref{lem:prestr}.

Assume that $C_2'$ violates A1, that is, $C_2'$ contains a cycle.
Since $C_2$ is cycle-free, there exists only one cycle in $C_2'$ and it contains $s$.
Adding $s$ to $C_2$ generates a cycle if and only if the endpoints of $s$ is connected in $C_2$.
Then, there exists a set $U \in \Pi(C_2)$ such that $\pts{s} \subseteq U$.
Since $\Pi(C_1) = \Pi(C_2)$, the set $U$ is also in $\Pi(C_1)$ and adding $s$ to $C_1$ generates a cycle in $C_1'$.
This contradicts that $C_1'$ is cycle-free.

Assume that $C_2'$ violates A2, that is, there exists a hidden component in $C_2'$.
Then, there exists a set $U \in \Pi(C_2)$ such that $U \subseteq \low{s}$.
Since $\Pi(C_1) = \Pi(C_2)$, $U$ is also in $\Pi(C_1)$ and the connected component containing $U$ in $C_1$ is hidden in $C_1'$, contradicting that $C_1'$ has no hidden components.
This finishes the proof of $C_2' \in \prestr$.

Next we show $C_1' \sim_{\phi} C_2'$.
Since $\sim_{\tau}$ is coherent by \cref{lem:tau}, $(W(C_1'), G(C_1'), B(C_1'), m(C_1')) = (W(C_2'), G(C_2'), B(C_2'), m(C_2'))$ holds.
Therefore, we have only to show $\Pi(C_1') = \Pi(C_2')$.

In the following, for a partition $\Pi$ of a set $E$ and a subset $T \subseteq E$,
$\Pi \restriction T$ denotes the \emph{restriction of $\Pi$ to $T$}, that is, $\Pi \restriction T := \inset{U \cap T}{U \in \Pi} \setminus \set{\emptyset}$.
For $C \in \prestr$, we define $\Pi^{+}(C)$ as the partition of $W(C) \cup G(C)$ such that every two points $x, y \in W(C) \cup G(C)$ is connected in $C$ if and only if they are in the same set of $\Pi^{+}(C)$.
In other words, $\Pi^{+}(C) = \Pi(C) \cup \inset{\set{p}}{p \in W(C)}$.
From the above discussion, the endpoints $a$ and $b$ of $s$ are in the different sets in $\Pi^{+}(C_1)$, otherwise, adding $s$ to $C_1$ generates a cycle in $C_1'$.
Let $U_a$ and $U_b$ be the sets containing $a$ and $b$ in $\Pi^{+}(C_1)$, respectively.
In $C_1'$, the vertices $x$ and $y$ are connected if and only if 1) $x$ and $y$ are in the same set in $\Pi^{+}(C_1)$ or 2) one of $x$ and $y$ is in $U_a$ and the other is in $U_b$.
Therefore, $\Pi(C_1') = ((\Pi^{+}(C_1) \setminus \set{U_a, U_b}) \cup \set{U_a \cup U_b}) \restriction G(C_1')$.

Since $\Pi(C_2) = \Pi(C_1)$ and $W(C_2) = W(C_1)$, it follows that $\Pi^{+}(C_2) = \Pi(C_2) \cup \inset{\set{p}}{p \in W(C_2)} = \Pi(C_1) \cup \inset{\set{p}}{p \in W(C_1)} = \Pi^{+}(C_1)$.
Thus, the sets $U_a$ and $U_b$ in $\Pi(C_1)$ are also contained in $\Pi(C_2)$, and the endpoints $a$ and $b$ of $s$ are also in $U_a$ and $U_b$, respectively.
Therefore, $\Pi(C_2') = ((\Pi^{+}(C_2) \setminus \set{U_a, U_b}) \cup \set{U_a \cup U_b}) \restriction G(C_2')$.
Since $\Pi^{+}(C_2) = \Pi^{+}(C_1)$ and $G(C_2') = G(C_1')$, we obtain $\Pi(C_2') = \Pi(C_1')$.

Lastly, we show $C_2 \xrightarrow{s} C_2'$.
To do this, it suffices to show that $s = \rex{C_2'}$ and $s \notin C_2'$.
The former holds because $\sim_{\tau}$ is coherent.
Assume that the latter is false: $s \in C_2'$.
In this case, $C_1 \sim_{\phi} C_2 = C_2' \sim_{\phi} C_1'$ holds.
Combining it with $C_1 \xrightarrow{s} C_1'$ means that $\prestr$ is not progressive on $\sim_{\phi}$.
However, since $\cfg \supseteq \prestr$ is progressive on $\sim_{\tau}$, the set $\prestr$ is progressive on $\sim_{\phi}$, which is a contradiction.

\section{Proof of \cref{theo:spanning_tree}}\label{app:spanning_tree}
The bound on the size of a combination graph is obtained from \cref{lem:size,lem:phi_size}.
To show the time bound, it suffices to show that, for every $C \in \prestr$ and $C'$ such that $C \xrightarrow{s} C'$, whether $C'$ is in $\prestr$ or not can be checked in $\bigO{n}$ time using only $\phi(C')$, not using the exact content of $C'$.

$C'$ is crossing-free because $s = \rex{C'}$.
From the proof of \cref{lem:phi_coherent}, $C'$ contains a cycle (i.e., violates A1) if and only if the two endpoints of $s$ are in the same set in $\Pi(C')$.
$C'$ contains a hidden component (i.e., violates A2) if and only if the there exists a set $U \in \Pi(C)$ such that $U \subseteq \low{s}$.
All the conditions can be checked in $\bigO{n}$ time.

\section{Proof of \cref{lem:psi_coherent}}


Let $C_1, C_2 \in \prescy$ be non-empty (otherwise, the proof is trivial) with $C_1 \sim_{\psi} C_2$ and assume that $C_1 \xrightarrow{s} C_1'$ holds for $C_1' \in \prescy$ and $s \in \sgm$.
Consider $C_2' := C_2 \cup \set{s}$.
We show $C_2' \in \prescy$, $C_1' \sim_{\psi} C_2'$, and $C_2 \xrightarrow{s} C_2'$, which implies the lemma.

To show $C_2' \in \prescy$, we show that $C_2'$ satisfies all the conditions from B1 to B3 in \cref{lem:prescy}.
Assume that $C_2'$ violates B1.
Then, there exists a point $p \in P$ such that $\degr{p}{C_2'} > 2$.
Since $C_1 \sim_{\psi} C_2$, we obtain $\degr{p}{C_1'} = \degr{p}{C_1} + 1 = \degr{p}{C_2} + 1 = \degr{p}{C_2'} > 2$, which contradicts that $C_1' \in \prescy$.

Assume that $C_2'$ violates B2.
Since $C_2$ satisfies B2, $s$ is the unique segment in $C_2'$ such that there exists a point $p \in \low
{s}$ with $\degr{p}{C_2'} < 2$.
By $p \in \low{s}$, $s$ is not incident to $p$.
Thus, adding $s$ to $C_1$ (or $C_2$) does not increase the degree of $p$.
Therefore, $\degr{p}{C_1'} = \degr{p}{C_1} = \degr{p}{C_2} = \degr{p}{C_2'} < 2$.
It contradicts that $C_1' \in \prescy$.

We can show that $C_2'$ satisfies B3 in the same way as the proof for A2 in \cref{lem:phi_coherent}.
It finishes the proof of $C_2' \in \prescy$.

Next, we show $C_1' \sim_{\psi} C_2'$, i.e., $\psi(C_1') = \psi(C_2)$.
First we show that $W(C_1') = W(C_2')$, $G(C_1') = G(C_2')$, and $B(C_1') = B(C_2')$.
Adding the segment $s$ to $C_1$ (or $C_2$) increases the degrees of the two endpoints of $s$ both by one and does not change the degrees of the other points.
Therefore, if a point $p$ is an endpoint of $s$, then $\degr{p}{C_1'} = \degr{p}{C_1} + 1 = \degr{p}{C_2} + 1 = \degr{p}{C_2'}$ holds.
Otherwise, $\degr{p}{C_1'} = \degr{p}{C_1} = \degr{p}{C_2} = \degr{p}{C_2'}$ holds.
Therefore, we have $W(C_1') = W(C_2')$, $G(C_1') = G(C_2')$, and $B(C_1') = B(C_2')$.

Next we show $m(C_1') = m(C_2')$.
To prove it, it suffices to show that $\rex{C_2'} = \rex{C_1'} = s$.
We first show that $s$ is extreme in $C_2'$.
Since $C \sim_{\psi} C'$, $\pts{C} = G(C) \cup B(C) = G(C') \cup B(C') = \pts{C'}$.
Combining this with the fact that $s$ is extreme in $C_1$, we obtain $\pts{s} \cap \low{C_2} \subseteq \pts{s} \cap \pts{C_2} = \pts{s} \cap \pts{C_1} = \emptyset$ and $\upp{s} \cap \pts{C_2} = \upp{s} \cap \pts{C_1} = \emptyset$.
Therefore, $s$ is extreme in $C_2'$.
Now, assume that $\rex{C_2'} = s' \neq s$.
Since both $s$ and $s'$ are extreme in $C_2'$, $s \prec s'$ holds, which means that $\rex{C_2} = s'$.
Combining it with $C_1 \sim_{\psi} C_2$ yields $s \prec \rex{C_1}$, which contradicts that $s = \rex{C_1'}$.

We can show $\mathcal{M}(C_1') = \mathcal{M}(C_2')$ in the same way as the proof of $\Pi(C_1') = \Pi(C_2')$ in \cref{lem:phi_coherent}.
It finishes the proof of $C_1' \sim_{\psi} C_2'$.

To show $C_2 \xrightarrow{s} C_2'$, it suffices to show that $s = \rex{C_2'}$ and $s \notin C_2$.
We have already proven the former.
To show the latter, it suffices to show that $\prescy$ is progressive on $\sim_{\psi}$.
For any $C, C' \in \prescy$ and $s \in \sgm$ such that $C \xrightarrow{s} C'$, $\sum_{p \in P} \degr{p}{C'} > \sum_{p \in P} \degr{p}{C}$ holds.
Therefore, $(W(C), G(C), B(C)) \neq (W(C'), G(C'), B(C'))$ holds, which means that $\prescy$ is progressive on $\sim_{\psi}$.

\section{Proof of \cref{theo:spanning_cycle}}
The bound on the size of a combination graph is obtained from \cref{lem:size,lem:psi_size}.
To show the time bound, it suffices to show that, for every $C \in \prescy$ and $C'$ such that $C \xrightarrow{s} C'$, whether $C'$ is in $\prescy$ or not can be checked in $\bigO{n}$ time.
$C'$ is crossing-free because $s = \rex{C'}$.
B1 can be checked in $\bigO{n}$ time.
To check if $C'$ satisfies the conditions from B1 to B3, it suffices to prune the cases that 1) there exists a point $p \in P$ such that $\degr{p}{C} > 2$, 2) there exists a point $p \in \low{s}$ such that $\degr{p}{C} < 2$, and 3) $\pts{s} \in \mathcal{M}$ and $\size{\mathcal{M}} \ge 2$.
All the conditions can be checked in $\bigO{n}$ time.


\end{document}